\documentclass{sig-alternate}

\usepackage{times}
\usepackage{graphicx}
\usepackage{booktabs}
\usepackage{tikz}
\usepackage{url}

\usepackage{algorithm}
\usepackage{algorithmic}

\usepackage{caption}
\DeclareCaptionType{copyrightbox}
\usepackage{subcaption}

\newtheorem{theorem}{Theorem}
\newtheorem{lemma}[theorem]{Lemma}
\newtheorem{cor}[theorem]{Corollary}

\newcommand{\p}{\mathbb{P}}  
\newcommand{\e}{\mathbb{E}}  

\newcommand{\m}[1]{\mathbf{#1}} 

\newcommand{\W}{\mathcal{W}} 

\begin{document}
%

\title{Learning multifractal structure in large networks}
%
%
%
%
%

\numberofauthors{3} 
%
\author{
%
%
\alignauthor
Austin R. Benson\\
       \affaddr{Stanford University}\\
       \affaddr{Institute for Computational and Mathematical Engineering}\\
       \email{arbenson@stanford.edu}
\alignauthor
Carlos Riquelme\\
       \affaddr{Stanford University}\\
       \affaddr{Institute for Computational and Mathematical Engineering}\\
       \email{rikel@stanford.edu}
\alignauthor Sven Schmit\\
       \affaddr{Stanford University}\\
       \affaddr{Institute for Computational and Mathematical Engineering}\\
       \email{schmit@stanford.edu}
}
\date{21 February 2013}

\maketitle
\begin{abstract}
  Generating random graphs to model networks has a rich history.
  In this paper, we analyze and improve upon the multifractal network generator (MFNG)  introduced by Palla \emph{et al}.
  We provide a new result on the probability of subgraphs existing in graphs generated with MFNG.
  From this result it follows that we can quickly compute moments of an important set of graph properties, such as the expected number of edges, stars, and cliques.
  Specifically, we show how to compute these moments in time complexity independent of the size of the graph and the number of recursive levels in the generative model.
  We leverage this theory to a new method of moments algorithm for fitting large networks to MFNG.
  Empirically, this new approach effectively simulates properties of several social and information networks.
  In terms of matching subgraph counts, our method outperforms similar algorithms used with the Stochastic Kronecker Graph model.
  Furthermore, we present a fast approximation algorithm to generate graph instances following the multifractal structure.
  The approximation scheme is an improvement over previous methods, which ran in time complexity quadratic in the number of vertices.
  Combined, our method of moments and fast sampling scheme provide the first scalable framework for effectively modeling large networks with MFNG.
\end{abstract}

\category{E.1}{Data structures}[Graphs and networks]



\section{Learning recursive graph structure}
\label{recursive}

Generative random graph models with recursive or hierarchical structure
are successful in simulating large-scale networks.
The recursive structure produces graphs
with heavy-tailed degree distribution and high clustering coefficient.
These random samples from recursive models are used to test algorithms, benchmark computer performance \cite{murphy2010introducing}, anonymizing and to understand the structure of networks.
Popular recursive and hierarchical models include Stochastic Kronecker Graphs (SKG, \cite{leskovec2010kronecker}),
Block Two-Level Erd\H{o}s-R\'{e}nyi (BTER, \cite{seshadhri2012community}), and the multifractal network generator (MFNG, \cite{palla2010multifractal, palla2011rotated}).
SKG is popular for several reasons.
Most importantly, the model captures degree distributions, clustering coefficients, and diameter.
There are several methods for fitting SKG parameters to simulate a target network.
The approaches include
maximum likelihood estimation (the KronFit algorithm, \cite{leskovec2010kronecker, leskovec2007scalable})
and a method of moments \cite{gleich2012moment}.
Maximum likelihood estimation is also used for the Multiplicative Attribute Graph model
\cite{kim2010multiplicative}, and
a simulated method of moments is used for mixed Kronecker product graph models \cite{moreno2013network, moreno2013learning}.
Finally, SKG produces graph samples in time complexity $\mathcal{O}(|E|\log(|V|))$
rather than $\mathcal{O}(|V|^2)$, where $E$ and $V$ are the edge and vertex sets
 of the graph.
On the other hand, SKG is constrained by a rather strong assumption on the
relationship between the number of recursion levels and the number of nodes in the graph.
Specifically, the number of recursive levels is $\lceil \log(|V|) \rceil$.

MFNG decouples the relationship between the recursion depth and the number of nodes,
and also naturally handles graphs where $|V|$ is not a power of two.
While there are ad-hoc methods for SKG when $|V|$ is not a power of two,
all analysis in the literature make the assumption.
We do not assume that $|V|$ is a power of two in our analysis in Section~\ref{sec:theory}.
For these reasons, MFNG is a more flexible model than SKG.
However, two issues are a barrier to making MFNG a practical model.
First, results for fitting graphs to the MFNG model have been
extremely limited.
Current procedures can only match a single graph property,
such as the number of nodes with degree $d$.
Second, to our knowledge, all MFNG sampling techniques are
$\mathcal{O}(|V|^2)$ algorithms,
making the generation of large graphs infeasible.

In this paper, we address both issues with MFNG and demonstrate that
can be a better alternative to the more popular SKG.
First, in Section~\ref{sec:theory}, we show how to compute several
key properties of MFNG (e.g., expected number of edges, triangles, stars, etc.)
with computational complexity independent of $|V|$ and the recursion depth.
Second, we develop a method of moments algorithm in Section~\ref{sec:mom}
to fit networks to MFNG.
The theory we develop in Section~\ref{sec:theory} makes this method extremely efficient and computationally tractable.
We test our new method of moments algorithm on synthetic data and large social and information networks.
In Section~\ref{sec:synthetic}, we show that our algorithm can identify model parameters in synthetic graphs sampled from MFNG,
and in Section~\ref{sec:learning}, we see that our algorithm can match the number of edges, wedges, triangles, $4$-cliques, $3$-stars, and $4$-stars in large networks.
In Section~\ref{sec:fast}, we provide a heuristic fast approximate sampling
scheme to randomly sample MFNG with complexity $\mathcal{O}(|E|)$.

\section{Overview of MFNG}
\label{sec:mfng}

MFNG is a recursive generative model based on a generating measure, $\W_k$.
The measure $\W_k$ consists of an $m$-vector of lengths $\ell$ with
$\sum_{i=1}^{m}\m \ell_i = 1$
and a symmetric $m \times m$ probability matrix $\m P$.
The subscript $k$ is the number of recursive levels, which we will subsequently explain.
In this paper, we refer to the $m$ indices of $\ell$ as \emph{categories}.
Also, since the measure is completely characterized by $\m P$, $\ell$, and $k$, we
write $\W_k(\m P, \ell)$ to explicitly describe the full measure.

An undirected graph $G = (V, E)$ is distributed according to $\W_k(\m P, \ell)$
if it is generated by the following procedure:
\begin{enumerate}
\item Partition $[0, 1]$ into $m$ subintervals of length $\ell_i$, $i = 1, \ldots, m$.
Recursively partition each subinterval $k$ times into $m$ pieces, using the relative lengths $\ell_i$.
This creates $m^k$ intervals $\ell_{i_1, \ldots, i_k}$ of length $\prod_{r=1}^{k} \ell_{i_r}$ such that
$\displaystyle \sum_{i_1, \ldots, i_k} \ell_{i_1, \ldots, i_k} = 1$.
\item
Sample $N$ points uniformly from $[0, 1]$ and create the nodes $V = \{x_1, \ldots, x_N\}$.
Each node $x_i$ is identified by its $k$-tuple of categories $c(x_i) = (i_1, \ldots, i_k)$, based on its
position on $[0, 1]$
and the partitioning in Step 1.

\item
For every pair of nodes $x_i$ and $x_j$ identified by the $k$-tuple of categories
$c(x_i) = (i_1, \ldots, i_k)$ and $c(x_j) = (j_1, \ldots, j_k)$,
add edge $(x_i, x_j)$ to $G$ with probability $\prod_{r=1}^{k} p_{i_rj_r}$.
\end{enumerate}

While the generation is intricate, MFNG admits a geometric interpretation.
Consider first the partition of the unit square into $m^2$ rectangles according to the lengths $\ell$.
The rectangle in position $(q, s)$ has side lengths $\ell_q$ and $\ell_s$, $1 \le q, s \le k$.
The point $(x_i, x_j) \in [0, 1] \times [0, 1]$ lands in the unit square, inside some rectangle $R$ with side lengths $\ell_{i_1}$ and $\ell_{j_1}$.
The edge `survives' the first round with probability $p_{i_1, j_1}$.
In the next round, we recursively partition $R$ according to the lengths $\ell$.
The relative positions of $x_i$ and $x_j$ land the point in a new rectangle
with side lengths $\ell_{i_2}$ and $\ell_{j_2}$.
The edge survives the second round with probability $p_{i_2, j_2}$.
The process is repeated $k$ times and is illustrated in Figure~\ref{fig:mfng}.
If an edge survives all $k$ levels, then it is added to the graph.
\begin{figure}[ht]
\vskip 0.2in
\begin{center}
\centerline{\includegraphics[width=\columnwidth]{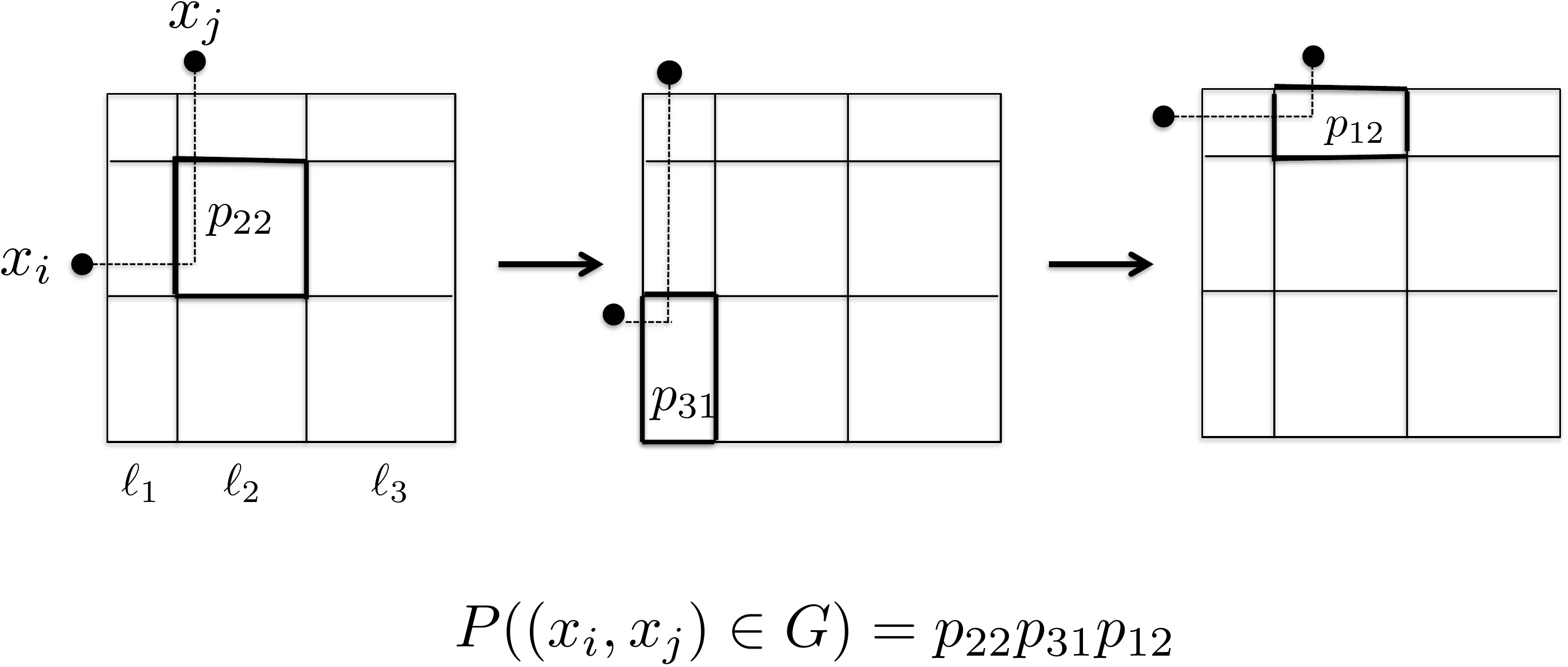}}
\caption{MFNG's recursive edge generation with $m = k = 3$.}
\label{fig:mfng}
\end{center}
\vskip -0.2in
\end{figure}


\section{Theoretical results}
\label{sec:theory}

    The original work on MFNG \cite{palla2010multifractal} shows how to compute the expected feature counts
    for graph properties by examining the entire expanded measure $\W_k(\m P, \ell)$.
    In other words, to count the features, the entire probability matrix of size $m^k \times m^k$ is formed.
    However, in some cases $m^k = \mathcal{O}(|V|)$ (see the examples in Section~\ref{sec:learning}), and computing $\mathcal{O}(|V|^2)$ probabilities is infeasible for large networks.
    Thus, current methods for counting and fitting features are intolerably expensive.
    Theorem~\ref{thm:main} shows that we can count many of the same features by only looking at the probability matrix $\m P$
    a constant number of times (independent of $|V|$).
    Hence, we are able to scale these computations to graphs with a large number of nodes.

\subsection{Decoupling of recursive levels}

We start with a lemma that shows how to decompose a generating measure $\W_k$ with $k$ recursive levels in $k$ measures with depth one.
This will make it easier to count subgraphs in Theorem~\ref{thm:main}.

\begin{lemma}
\label{thm:coupling}
    Consider generating measures $\W_1(\m P, \ell)$ and $\W_k(\m P, \ell)$, which are parameterized by the same probability matrix $\m P$
    and lengths $\ell$ but different recursion depths.
    Let graphs $H_1, \ldots, H_k \sim \W_1(\m P, \ell)$ be independently drawn, and also denote $H_i = (V, E_{i})$, with nodes labelled arbitrarily.
    Then the intersection graph $G = (V, \cap_{i=1}^{k}E_{i}) = (V, E_G) \sim \W_k(\m P, \ell)$.
\end{lemma}

\begin{proof}
    We prove the lemma by conditioning on the categories to which the nodes belong
    (recall that a category is the set of intervals that a node falls into at each level of the recursion). 
    Each node $u \in V$ is identified with some real number in $[0, 1]$.
    The probability that the $k$-tuple of categories corresponding to
    $u$ is $c(u) = (c_1, \ldots, c_k)$ in any graph $H \sim \W_k(\m P, \ell)$ is simply $\prod_{r=1}^{k}\ell_{c_r}$.
    By independence of the $H_i$, the probability that the same node $u$ is in the same categories $c_1, \ldots, c_k$
    in the graphs $H_1, \ldots, H_k$, respectively, is also  $\prod_{r=1}^{k}\ell_{c_r}$.

    Note that
    \begin{align}
        &  \p\left((u, v)\in E_G | c(u) = (c^u_1, \ldots, c^u_k), c(v) = (c^v_1, \ldots, c^v_k)\right) \nonumber \\
        = &\p\left((u, v)\in \cap_{i=1}^{k} E_{i} | c(u) = (c^u_1, \ldots, c^u_k), c(v) = (c^v_1, \ldots, c^v_k)\right) \nonumber \\
        = &\prod_{i=1}^{k}\p\left((u, v)\in E_{i} | c(u) = (c^u_1, \ldots, c^u_k), c(v) = (c^v_1, \ldots, c^v_k)\right) \nonumber \\
        = &\prod_{i=1}^{k}\p\left((u, v)\in E_{i} | [c(u)]_i = c^u_i, [c(v)]_i = c^v_i \right) \nonumber \\
        = &\prod_{i=1}^k p_{c^u_i,c^v_i}. \nonumber
    \end{align}
    In the first equality, we use the definition of $E_G$;
    in the second and third equalities, we use the independence of the $H_i$;
    and in the final equality, we use the definition of $\W_1$.
    However, for any graph $G^{\prime} \sim \W_k(\m P, \ell)$,
    \begin{align}
        &  \p\left((u, v)\in G^{\prime} | c(u) = (c^u_1, \ldots, c^u_k), c(v) = (c^v_1, \ldots, c^v_k)\right) \nonumber \\
         =&\prod_{i=1}^k p_{c^u_i,c^v_i}. \nonumber
    \end{align}
\end{proof}

Figure~\ref{fig:coupling} illustrates Lemma~\ref{thm:coupling}.
Our main result is a straightforward consequence of this lemma.

\begin{figure}
\usetikzlibrary{arrows,positioning}
\begin{tikzpicture}

\node[draw,circle] (v2) at (-4.5,3.5) {};
\node[draw,circle] (v5) at (-4.5,2.5) {};
\node[draw,circle] (v4) at (-3.5,2.5) {};
\node[draw,circle] (v3) at (-3.5,3.5) {};
\node[draw,circle,label=above:{$H_1$}] (v1) at (-4,4.5) {};
\node[draw,circle] (v6) at (-2.5,3.5) {};
\node[draw,circle] (v7) at (-1.5,3.5) {};
\node[draw,circle] (v9) at (-2.5,2.5) {};
\node[draw,circle] (v8) at (-1.5,2.5) {};
\node[draw,circle,label=above:{$H_2$}] (v10) at (-2,4.5) {};
\node[draw,circle] (v11) at (-0.5,3.5) {};
\node[draw,circle] (v12) at (-0.5,2.5) {};
\node[draw,circle] (v13) at (0.5,2.5) {};
\node[draw,circle] (v14) at (0.5,3.5) {};
\node[draw,circle,label=above:{$H_3$}] (v15) at (0,4.5) {};
\node[draw,circle] (v20) at (2.5,2.5) {};
\node[draw,circle] (v21) at (3.5,2.5) {};
\node[draw,circle] (v22) at (2.5,3.5) {};
\node[draw,circle] (v19) at (3.5,3.5) {};
\node[draw,circle,label=above:{$G$}] (v18) at (3,4.5) {};

\draw  (v1) edge (v2);
\draw  (v2) edge (v3);
\draw  (v3) edge (v4);
\draw  (v5) edge (v4);
\draw  (v1) edge (v3);
\draw  (v6) edge (v7);
\draw  (v8) edge (v9);
\draw  (v10) edge (v6);
\draw  (v10) edge (v7);
\draw  (v11) edge (v12);
\draw  (v12) edge (v13);
\draw  (v13) edge (v14);
\draw  (v14) edge (v15);

\draw  (v18) edge (v19);
\draw  (v20) edge (v21);
\draw  (v15) edge (v11);
\draw  (v18) edge (v22);

\node (v17) at (2,3.5) {};
\node (v16) at (1,3.5) {};
\draw [-triangle 90,] (v16) edge node[auto] {$\cap$} (v17);

\end{tikzpicture}
\caption{Illustration of Lemma~\ref{thm:coupling}.
         If three graphs $H_1, H_2,$ and $H_3$ are generated from $\W_1(\m P, \ell)$,
         then their intersection $G$ follows the distribution of $\W_3(\m P, \ell)$.}
\label{fig:coupling}
\end{figure}
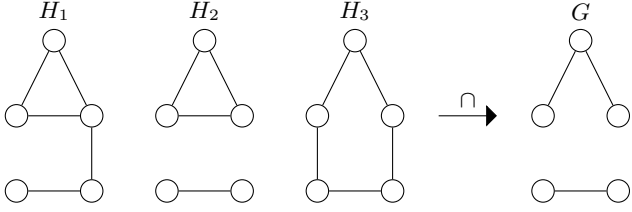

\begin{theorem}
\label{thm:main}
    Let $\W_1(\m P, \ell)$ and $\W_k(\m P, \ell)$ be generating measures defined by the same probabilities $\m P$ and lengths $\ell$
    but with different recursion depths.
    Consider $k$ multifractal graphs $H_i = (V, E_i)$ generated independently from $\W_1(\m P, \ell)$
    and a multifractal graph $G = (V, E_G)$ generated from $\W_k(\m P, \ell)$.
    For any event $A$ on $G$ that can be written as $A = \{S \subset E_G \}$, where $S \subset \{(i,j) : i,j \in \{1,\ldots,n\}, i<j\}$,
    \begin{equation}
        \p_{\W_k}(A) = \p_{\W_1}(A)^k \nonumber.
    \end{equation}
\end{theorem}

\begin{proof}
    \begin{align}
    \p_{\W_k}(A) &= \p_{\W_k}(s \in E_G, \forall s \in S) \nonumber \\
    &= \p_{(\W_1)^k}(s \in E_{i}, \forall s \in S, \forall i \in \{1,\ldots,k\}) \nonumber\\
    &= \prod_{i=1}^k \p_{(\W_1)^k}(s \in E_{i}, \forall s \in S) \nonumber \\
    &= \p_{(\W_1)^k}(s \in E_{1}, \forall s \in S)^k \nonumber \\
    &= \p_{\W_1}(A)^k. \nonumber
    \end{align}
\end{proof}

In other words, the probability that a subset of the edges exists if the graph is drawn from $\W_k$ is the $k$-th power of the probability that these
edges exist if the graph is drawn from $\W_1$.
The condition that $A$ can be written as $A = \{S \subset E \}$ is subtle.
It states that Theorem~\ref{thm:main} holds if we can specify a subset of the edges that must be present in the graph.
We can be indifferent about certain edges, but we cannot specify that an edge is \emph{not} present in the graph.

We can now easily compute the moments of subgraph counts, such as the number of edges, triangles, and larger cliques in MFNG.
The following corollary shows how to use Theorem~\ref{thm:main} for these calculations.
for graphs generated by MFNG.
\begin{cor}\label{cor:edges}
The expected number of edges $|E|$ in a graph sampled from MFNG is
\begin{equation}
    \e[|E|] = \binom{n}{2} s^k,
\end{equation}
where
\begin{equation}
    s = \sum_{i,j \in [m]} p_{ij} \ell_i \ell_j.
\end{equation}
\end{cor}

\begin{proof}
Let $u$ and $v$, $u \neq v$, be two random nodes of $G$.
Let $A$ denote the event $A = \{ (u, v) \in E \}$,
and we define $A_i$ to denote the analogous event restricted to $H_i$ in the multifractal generator.
By Theorem~\ref{thm:main}, we have that
$$ \p(A) = \prod_{i=1}^k \p(A_i) = \p(A_1)^k.  $$
Now that we can restrict ourselves to $A_1$,
\begin{align}
 \p(A_1) &= \sum_{i,j \in [m]}  \p(A^{(1)} | c^u_1 = i, c^v_1 = j) \ \p(c^u_1 = i, c^v_1 = j) \\
 &= \sum_{i,j \in [m]}  p_{ij} \ \p(c^u_1 = i) \p(c^v_1 = j) \\
 &= \sum_{i,j \in [m]}  p_{ij} \ \ell_i \ell_j = s.
\end{align}
We conclude that
\begin{equation}
\p(A) = \p((u, v) \in E) = s^k.
\end{equation}
The expected number of edges is then given by
\begin{equation}
\e[|E|] = \binom{n}{2} \ s^k.
\end{equation}
\end{proof}

\begin{cor}
\label{cor:moments}
Graphs sampled from MFNG also have the following moments.
The expected number of
$d$-stars~\footnote{
A $d$-star is a graph with $d+1$ vertices and $d$
edges that connect the first node to all other vertices.}
$S_d$ is:
\begin{equation}
\e[S_d] = n \binom{n-1}{d}
\left(
\sum_{i_1, \ldots, i_{d+1} \in [m]}
\prod_{j=2}^{d+1}p_{i_1i_j}
\prod_{j=1}^{d+1}\ell_{i_j}
\right)^k. \nonumber
\end{equation}
In particular, the expected number of wedges ($2$-stars) is
\begin{equation}
\e[S_2] = n \binom{n-1}{2}
\left(
\sum_{i_1, i_2, i_3 \in [m]} p_{i_1i_2}p_{i_1i_3} \ \ell_{i_1}\ell_{i_2}\ell_{i_3}
\right)^k. \nonumber
\end{equation}
The variance $\sigma_E =  \text{Var}(|E|)$ of the number of edges is
\begin{align}
\sigma_E &= \binom{n}{2} s^k \left(1- \binom{n}{2} s^k \right) + 2 \ \e[S_2] + {n \choose 2}{n-2 \choose 2} s^{2k} \nonumber,
\end{align}
where $s$ is the same as in Corollary~\ref{cor:edges}.

The expected number of $t$-cliques~\footnote{
A $t$-clique is a graph with $t$ vertices where
every possible edge between the vertices exists.}
$C_t$ is
\begin{equation}
\e[C_t] = \binom{n}{t} \ s_t^k,
\end{equation}
where
\begin{equation}
s_t := \sum_{i_1,\dots,i_t \in [m]}
\left(
 \prod_{\substack{ j,q \in [t] \\ j < q }} p_{i_ji_q}
\right)
 \ell_{i_1} \ell_{i_2} \cdots \ell_{i_t}.
\end{equation}
In particular, the expected number of triangles ($3$-cliques) is:
\begin{equation}
\e[C_3] = \binom{n}{3} \left(\sum_{i, j, t \in [m]} p_{ij} p_{it} p_{jt} \ \ell_i \ell_j \ell_t\right)^k.
\end{equation}
Finally, the expected number of nodes with degree $d$, $E_d$,
satisfies $\e[E_{|V|-1}] = \e[S_{|V|-1}] $ and
\begin{equation}
\e[E_d] = \e[S_d] - \sum_{i = d+1}^{|V|-1} \binom{i}{d} \e[E_{i}].
\end{equation}
\end{cor}

\begin{proof}
The proofs follow the same patterns as of the proof of Corollary~\ref{cor:edges}.
We include the proofs in supplementary material online
\footnote{\url{http://stanford.edu/~arbenson/mfng.html}}.
\end{proof}
These are some examples of properties for which we can compute the \emph{exact} expectation.
However, we can also compute useful approximations.
For a given measure $\W_k$, we could empirically compute the value of $\e[C_t]$ for each $t$ until we find $\e[C_{t^*}] \ge 1 > \e[C_{t^*+1}]$,
which is a good estimator of the expected maximum clique size.
However, a concentration result is still needed to claim that the clique number will be in a small neighborhood of $t^*$ with high probability.

Finally, we note that there are graph properties which will certainly not translate to this theoretical framework.
Let $\mu(G)$ to be the chromatic number of $G$,
i.e., the smallest number of colors needed to color the vertices such that vertices connected by an edge are not the same color.
Suppose we want to compute $\p(\mu(G) < 10)$.
If the theorem is used directly, then the result is $\p(\mu(G) < 10) = \p(\mu(H_1) < 10)^k$.
But $\p(\mu(G) < 10) \ge \p(\mu(H_1) < 10)$ since taking the intersection of graphs can only reduce the chromatic number.
In this case, $\p(\mu(G) < 10)$ cannot be written as an event on the subset of the edges of the graph.
Hence, the assumptions of the theorem are violated.

%

%

\section{Method of moments learning algorithm}
\label{sec:mom}

From now on, we change gears and look at how we can use the theory laid out above to fit multifractal measures to real networks.
Given a graph $G$, we are interested in finding a probability matrix $\m P$, a set of lengths $\ell$, and a recursion depth $k$,
such that graphs generated from the measure $\W_k(\m P, \ell)$ are similar to $G$.
The theoretical results in Section~\ref{sec:theory} make it simple to compute moments for MFNG, so a method of moments is natural.
In particular, given a set of desired features counts $F_i$ (such as number of edges, $2$-stars, and triangles), we seek to solve the following optimization problem:

\begin{equation}
\begin{aligned}
& \underset{\m P, \ell, k}{\text{minimize}}
& & \sum_{i}\frac{|F_i - \mathbb{E}_{\W_k}[F_i]|}{F_i}  \\
& \text{subject to}
& & 0 \le p_{ij} = p_{ji} \le 1,     & 1 \le i \le j \le m  \\
& & & 0 \le \ell_{i} \le 1, & 1 \le i \le m \\
& & & \sum_{i=1}^{m}\ell_i = 1
\end{aligned}
\label{eq:opt_MFNG}
\end{equation}

If certain features are more important to fit,
then the objective function can be generalized to include weights, i.e.
\[
\sum_{i}  \frac{w_i|F_i - \mathbb{E}_{\W_k}[F_i]|}{F_i},
\]
for $w_i \ge 0$.
For simplicity of our numerical experiments,
we only use an unweighted objective in this paper.
Similar objective functions were proposed for SKG \cite{gleich2012moment}
and for mixed Kronecker product graph models \cite{moreno2013learning}.
In Section~\ref{sec:experimental}, we see that the simple objective function works well on synthetic and real data sets.

\subsection{Desired features}
\label{sec:desired}
We want to model real world networks well, while also being computationally feasible.
Theorem~\ref{thm:main} shows that, given a generating measure $\W_k(\m P, \ell)$, we can quickly compute moments of several feature counts.
However, we are also interested in global graph properties such as degree distribution and clustering coefficient.
These properties are not covered by our theoretical results.
To this end, we compute the expected number of
$d$-stars and
$t$-cliques\footnote{From now on we implicitly mean counting subgraphs if we say counting $d$-stars or $t$-cliques.},
and use those as a proxy.
If the number of $d$-stars and $t$-cliques are similar,
then we expect the degree distribution and clustering to be also similar.
In particular, the global clustering coefficient is three times the ratio of the number of triangles ($3$-cliques) to the number of wedges ($2$-stars) in the graph.
In Section~\ref{sec:learning}, we show that matching star and clique subgraph counts in social and information networks leads to a generating measure that produces graphs with a similar degree distribution.

\subsection{Solving the optimization problem}
Optimization problem~(\ref{eq:opt_MFNG}) is not trivial to solve,
as there are many local minima and some of them turn out to be very poor.
On the other hand, given the feature counts of a graph,
running a standard optimization solver such as \texttt{fmincon} in Matlab,
finds such a critical point quickly:
we only have to fit $m^2 + m + 1$ variables.
Typically, $m$ is two or three.
Thus, we solve the optimization problem with many random restarts and use the best result.
While there are more sophisticated methods,
this method works on several practical examples (see Section~\ref{sec:experimental}).

\section{Fast sampling for sparse graphs}
\label{sec:fast}
In this section, we discuss a heuristic method for generating sample graphs following the multifractal measure that is effective when the graph to be generated is sparse, i.e. has relatively few edges.
This is important because the naive sampling method takes $\mathcal{O}(|V|^2)$ time---it considers the edge for every pair of nodes in the graph.
The fast heuristic algorithm is inspired by the ``ball-dropping'' scheme for SKG (see Section 3.6 of \cite{leskovec2010kronecker}).
However, due to the stochastic nature of the location of the nodes, our algorithm is not exact and merely a heuristic, unlike in the SKG case.
The speed-up is obtained by fixing the number of edges in advance and only considering $\mathcal{O}(|E|)$ pairs of vertices.
We will demonstrate that our sampling algorithm runs in time $\mathcal{O}(|E|\log(|V|))$ time.
The pseudo-code is given in Algorithm~\ref{alg:fast}.
In the subsequent sections, we give the details of the algorithm and briefly discuss the performance.

\subsection{The algorithm}
In order to avoid looping over all pairs of nodes, we fix the number of edges.
The number of edges is determined by sampling a normal random variable with mean $\e[|E|]$ and variance $\sigma_E$,
as provided by Corollaries~\ref{cor:edges}~and~\ref{cor:moments}.
Since the number of edges is a sum of Bernoulli trials, the normal approximation is accurate.

Now that we have selected the number of edges to add, it is time to add edges to $E$.
Because node locations are random (i.e., every node has a random category), it is nontrivial to select a candidate edge.
This contrasts with SKG, where the edge probabilities for a given node is deterministic.
Because of the stochastic locations of the nodes in MFNG, our fast sampling algorithm is only an approximation.
The algorithm proceeds by selecting node categories level by level, for each each edge.
To select categories, we sample an index $(c, c')$ of a matrix $\m Q$:
\[
  \m Q_{ij} = p_{ij} \ell_i \ell_j.
\]
The sampling is done proportional to the entries in $\m Q$.
The matrix $\m Q$ reflects the relative probability mass corresponding to an edge falling into those categories.
In other words, it denotes the probability of selecting the categories $c_1$ and $c'_1$ at a given level \emph{and} the edge surviving the level.
The category sampling is performed $k$ times, one for each level of recursion.
This gives two $k$-tuples of categories: $c = (c_1, \ldots, c_k)$ and $c' = (c_1', \ldots, c_k')$.

Now we want to add an edge between nodes $u$ and $u'$ that have the categories $c$ and $c'$.
However, we have to be careful about the number of nodes that have the same category.
We can think of the category pair $(c, c')$ as a box $B$ on the generating measure.
Consider two boxes $B_1$ and $B_2$ and suppose that both have the same area in the unit square,
and the probability between potential boxes in $B_1$ and $B_2$ is the same.

A simple example is the following case:
\begin{itemize}
\item $k = m = 2$
\item $p_{11} = p_{22} = p_{12} = 0.5$
\item $\ell_1 = 0.5$, $\ell_2 = 0.5$
\end{itemize}
The edge probabilities in any two boxes $B_1$ and $B_2$ in the measure are the same, and the probabilities
of selecting either box (from sampling the $\m Q$ matrix) are the same.
However, because of the randomness categories for nodes, there may be 10 node pairs in $B_1$ and only one node pair in $B_2$.
If we simply pick a node pair at random from a box, the probability of connecting the node pair in $B_2$ is much higher
than the probability of connecting node pairs in $B_1$.

To overcome this discrepancy, we take into account the difference between the \emph{expected} number of nodes pairs in a box and the \emph{actual} number of node pairs in a box.
Note that the joint distribution of nodes is Multinomial$(n; l_1, l_2, \ldots, l_{m^k})$ where $l_i$ denotes the length of interval $i$ (after recursive expansion).
Let $p_{c,c'}$ be the edge probability in the box corresponding to the category pair $(c, c')$.
Let the box's sides have lengths $l$ and $l'$.
Using standard properties of the Multinomial distribution,
\[
   n_{c,c'} = \left\{
     \begin{array}{lr}
       |V| (|V| l^2 - l^2 + l) & \text{if } c = c' \\
       |V| (|V| - 1) l l' & \text{if } c \neq c'
     \end{array}
   \right.
\]
Finally, we sample
\[
e_{\text{to add}} \sim \text{Poisson}\left(\frac{n_{c,c'}}{\lambda |V_c|||V_{c'}|}\right),
\]
where
\[
V_c = \{v \in V | \text{category of v is } c\}.
\]
We then add $e_{\text{to add}}$ edges to the box $(c, c')$.
Thus, if there are many more node pairs in a box than expected, we add more edges to the box.

There are a couple of details we have swept under the rug.
First, we haven't discussed what to do if the box $(c, c')$ is empty.
In this case, we simply re-sample $c$ and $c'$.
In practice, this does not occur too frequently.
Second, we have introduced some dependence between edges, and MFNG samples edges independently.
For this reason, we use have included the accuracy factor $\lambda$.
By increasing $\lambda$, the sampling takes longer, but there is less edge dependence.

\subsection{Performance}
The speedup achieved by this fast approximation algorithm really depends on the type of graph.
We trade an $\mathcal{O}(|V|^2)$ algorithm for an algorithm that takes $\mathcal{O}(|E| \log |V|)$ time if there are no rejected tries due to empty boxes, edges that are already present, etc.
In the case that the graph is sparse and $k <\approx \log_m n$, this is fine.
However, for denser graphs, this fast method will actually turn out to be slower.
To arrive at a complexity of $\mathcal{O}(|E| \log |V|)$ we note that it takes $\mathcal{O}(|V| \log |V|)$ time to compute the categories for each $u \in V$.
Then, assuming that the number of retries is small, the while loop of Algorithm~\ref{alg:fast} is executed $\mathcal{O}(|E|)$ times,
each taking $\mathcal{O}(k) = \mathcal{O}(\log |V|)$ steps.
Therefore, in total, the algorithm has complexity $\mathcal{O}( |E| \log |V|)$.

\begin{algorithm}[t]
   \caption{Fast approximate sampling algorithm}
   \label{alg:fast}
\begin{algorithmic}[1]
    \STATE {\bfseries Input:} Generating measure $\W_k(\m P, \ell)$,
      accuracy factor $\lambda$
    \STATE {\bfseries Output:} Graph $G$ with distribution approximately $\W_k(\m P, \ell)$.
    \STATE Add $|V|$ nodes by uniformly sampling on $[0,1]$ and assigning the proper categories to each node.
    \STATE Set $V_c = \{v \in V | \text{category of v is } c\}$ for each category $c$.
    \STATE Fix number of candidate edges $|E| = \lfloor \mathcal{E} \rfloor$,
            where $\mathcal{E} \sim N(\mu_{|E|}, \sigma_{|E|})$.
    \STATE Compute $\m Q$, where $Q_{ij} = p_{ij} \ell_i \ell_j$ for $1 \le i, j, \le m$
    \STATE Set $e_{\text{global}} = 0$
    \WHILE{$e_{\text{global}} < |E|$}
      \FOR{$h=1$ {\bfseries to} $k$}
          \STATE Pick category $c_h, c_h'$ independently and with probability proportional to $\m Q_{c_h, c_h'}$
      \ENDFOR
      \STATE $c = (c_1, \ldots, c_k)$, $c' =  (c'_1, \ldots, c'_k)$.
      \STATE Set $l$, $l'$ to lengths of interval corresponding to $c$, $c'$
      \IF{$|V_c|||V_{c'}| \neq 0$}
        \IF{$c = c'$}
          \STATE $n_{c, c'} = |V| (|V| l^2 - l^2 + l)$
        \ELSE
          \STATE $n_{c,c'} = |V| (|V| - 1) l l'$
        \ENDIF
        \STATE Draw $e_{\text{to add}} \sim $Poisson$\left(n_{c,c'} / (\lambda |V_c|||V_{c'}|)\right)$
        \STATE Set k = 0
        \STATE Set $e_{\text{local}} = 0$
        \WHILE{$e_{\text{local}} < e_{\text{to add}}$ and $k < \max_k$}
          \STATE Pick $u \in V_c$ and $v \in V_{c'}$ uniform at random.
          \IF{$(u,v) \notin E$ and $u \neq v$}
            \STATE Add $(u,v)$ to $E$
            \STATE Set $e_{\text{local}} = e_{\text{local}} + 1$
          \ENDIF
          \STATE Set $k = k + 1$
        \ENDWHILE
        \STATE $e_{\text{global}} = e_{\text{global}} + e_{\text{local}}$
      \ENDIF
    \ENDWHILE
    \STATE Return $G = (V,E)$
\end{algorithmic}
\end{algorithm}

\section{Experimental results}
\label{sec:experimental}
In the next sections, we demonstrate the effectiveness of our approach to model networks.
First, we show that our method is able to recover the multifractal structure if we generate synthetic graphs following the MFNG paradigm.
Thereafter, we consider several real-world networks and compare the performance of our method to popular methods using SKG.

\begin{figure}[tb]
\usetikzlibrary{arrows,positioning}
\centering
\begin{subfigure}{0.12\textwidth}
  \centering
   \resizebox{\linewidth}{!}{

     \begin{tikzpicture}
       \node[draw,circle] (v2) at (-4.5,3.5) {};
       \node[draw,circle] (v3) at (-3.5,3.5) {};
       \node[draw,circle,label=above:{$S_2$}] (v1) at (-4,4.5) {};

       \draw  (v1) edge (v2);
       \draw  (v1) edge (v3);

     \end{tikzpicture}

   }
\end{subfigure}
\begin{subfigure}{0.12\textwidth}
  \centering
   \resizebox{\linewidth}{!}{
     \begin{tikzpicture}
       \node[draw,circle] (v2) at (-4.5,3.5) {};
       \node[draw,circle] (v3) at (-3.5,3.5) {};
       \node[draw,circle] (v4) at (-4,3.5) {};
       \node[draw,circle,label=above:{$S_3$}] (v1) at (-4,4.5) {};

       \draw  (v1) edge (v2);
       \draw  (v1) edge (v3);
       \draw  (v1) edge (v4);
     \end{tikzpicture}
   }
\end{subfigure}
\begin{subfigure}{0.17\textwidth}
  \centering
   \resizebox{\linewidth}{!}{
     \begin{tikzpicture}
       \node[draw,circle] (v2) at (-4.75, 3.5) {};
       \node[draw,circle] (v3) at (-4.25, 3.5) {};
       \node[draw,circle] (v4) at (-3.75, 3.5) {};
       \node[draw,circle] (v5) at (-3.25, 3.5) {};
       \node[draw,circle,label=above:{$S_4$}] (v1) at (-4,4.5) {};

       \draw  (v1) edge (v2);
       \draw  (v1) edge (v3);
       \draw  (v1) edge (v4);
       \draw  (v1) edge (v5);
     \end{tikzpicture}
   }
\end{subfigure}
\begin{subfigure}{0.11\textwidth}
  \centering
   \resizebox{\linewidth}{!}{
     \begin{tikzpicture}
       \node[draw,circle] (v2) at (3,3) {};
       \node[draw,circle] (v3) at (3.75,3.75) {};
       \node[draw,circle,label=above:{$C_3$}] (v1) at (3,3.75) {};

       \draw  (v1) edge (v2);
       \draw  (v1) edge (v3);
       \draw  (v2) edge (v3);
     \end{tikzpicture}
   }
\end{subfigure}
\begin{subfigure}{0.11\textwidth}
  \centering
   \resizebox{\linewidth}{!}{
     \begin{tikzpicture}
       \node[draw,circle] (v2) at (3,3) {};
       \node[draw,circle] (v3) at (3.75,3.75) {};
       \node[draw,circle] (v4) at (3.75,3) {};
       \node[draw,circle,label=above:{$C_4$}] (v1) at (3,3.75) {};

       \draw  (v1) edge (v2);
       \draw  (v1) edge (v3);
       \draw  (v1) edge (v4);
       \draw  (v2) edge (v3);
       \draw  (v2) edge (v4);
       \draw  (v3) edge (v4);
     \end{tikzpicture}
   }
\end{subfigure}
\caption{$d$-stars and $t$-cliques features that are counted in the experiments in Section~\ref{sec:experimental}.}
\label{fig:features}
\end{figure}
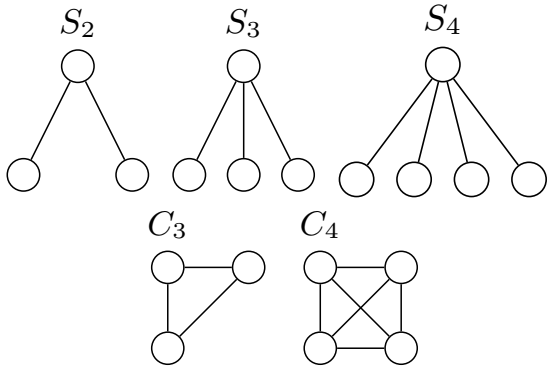

\subsection{Identifiability and learning on synthetic networks}
\label{sec:synthetic}

\begin{figure*}[tb]
\vskip 0.2in
\begin{center}
\includegraphics[width=3.3in]{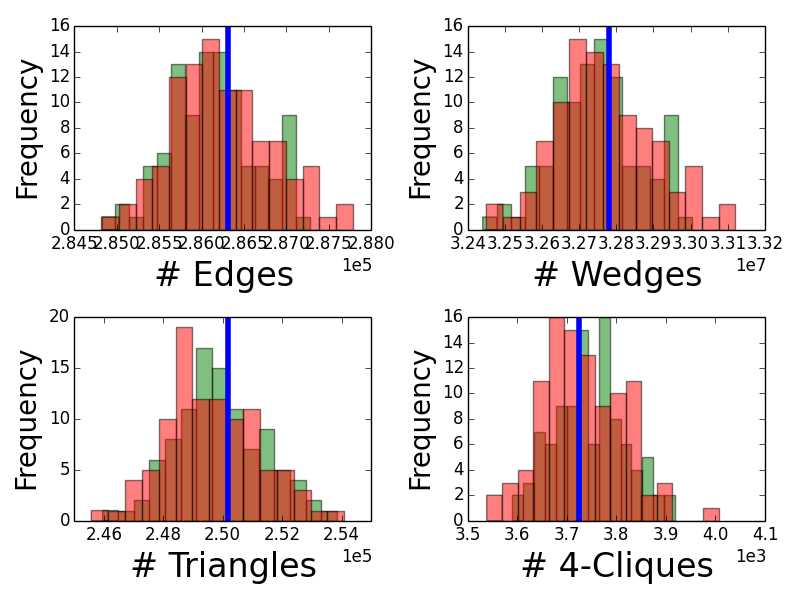}
\includegraphics[width=3.3in]{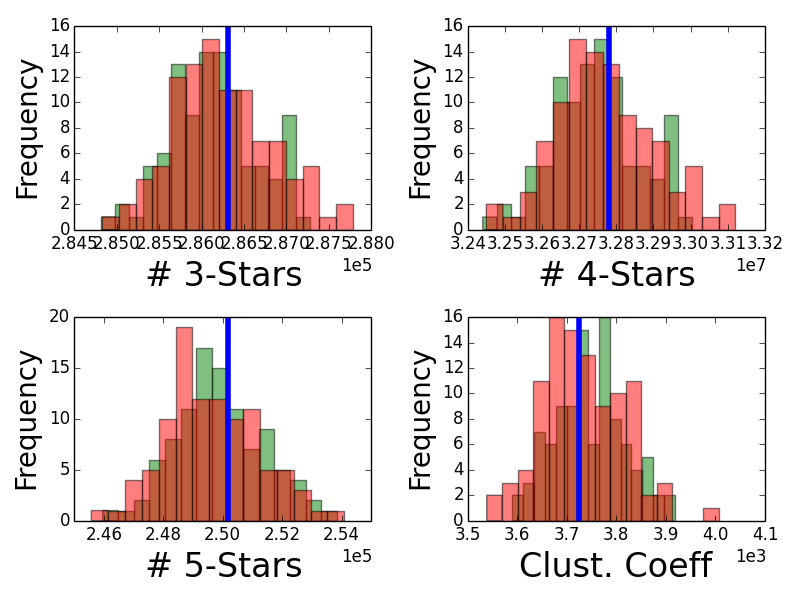}
\caption{
Empirical distributions of feature counts and clustering coefficient for the original MFNG (green) and the retrieved MFNG (red) found with the method of moments algorithm.
The blue line is the feature count from the single sample of the original MFNG used in the method of moments.
In this case, the original MFNG followed an Erd\H{o}s-R\'{e}nyi model.
The original and retrieved measures produce similar distributions.
}
\label{fig:hist_ER_recovery}
\end{center}
\vskip -0.2in
\end{figure*}

\begin{figure*}[tb]
\vskip 0.2in
\begin{center}
\includegraphics[width=3.3in]{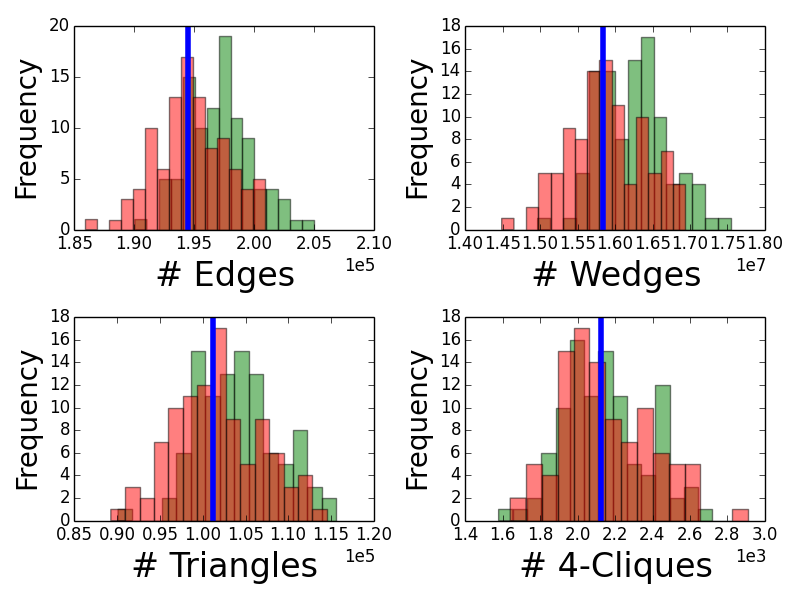}
\includegraphics[width=3.3in]{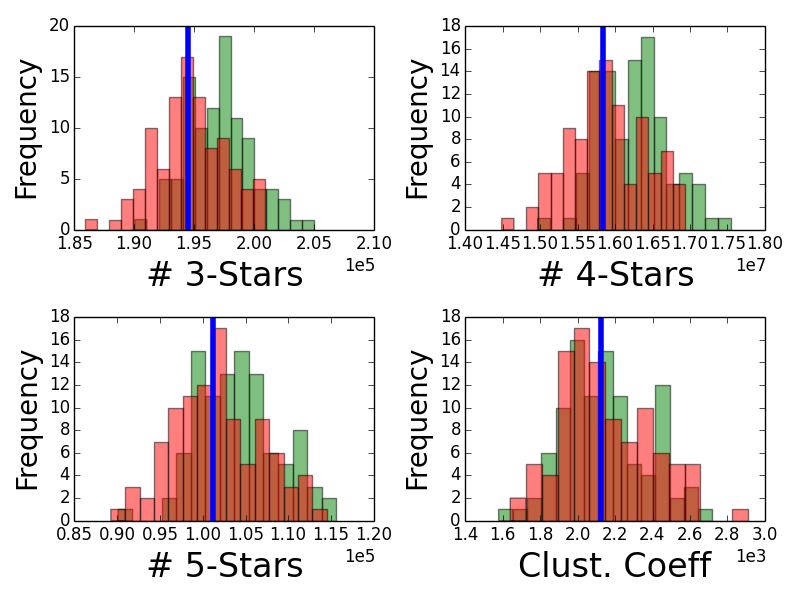}
\caption{
Empirical distributions of feature counts and clustering coefficient for the original MFNG (green) and the retrieved MFNG (red) found with the method of moments algorithm.
The blue line is the feature count from the single sample of the original MFNG used in the method of moments.
The original and retrieved measures produce similar distributions.
}
\label{fig:hist_mfng_recovery}
\end{center}
\vskip -0.2in
\end{figure*}

\begin{figure*}[tb]
\vskip 0.2in
\begin{center}
\includegraphics[width=3.3in]{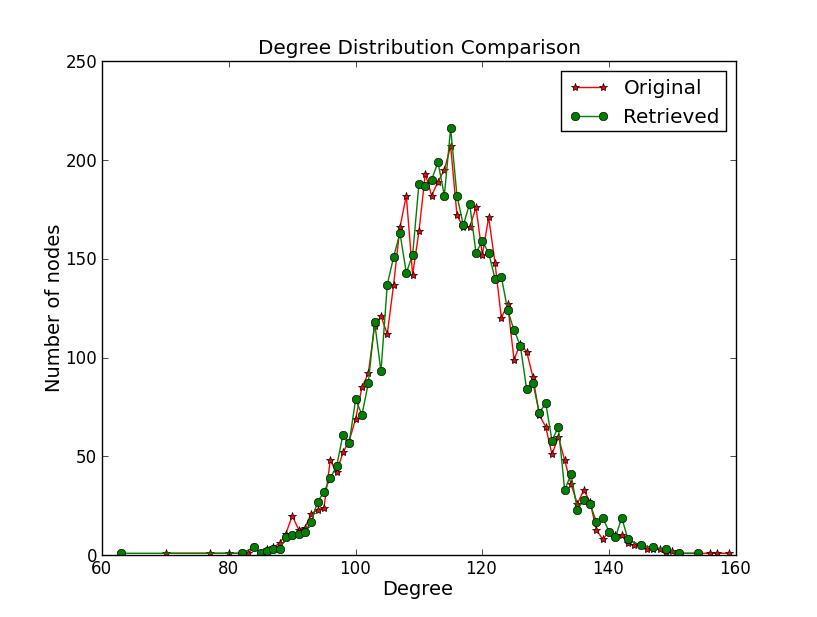}
\includegraphics[width=3.3in]{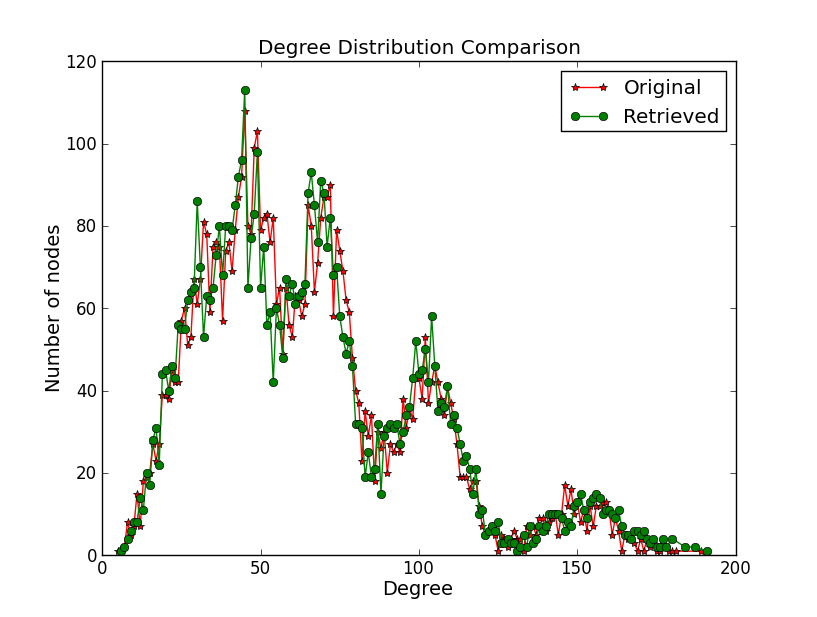}
\caption{
On the left, degree distribution of graphs generated according to the original Erd\H{o}s-R\'{e}nyi measure (red) given in Table \ref{tab:ER_recovery} and the retrieved measure (green). On the right, degree distribution of graphs generated according to the original measure (red) described in Table \ref{tab:mfng_recovery} and the retrieved measure (green).
The retrieved measure was found by the method of moments algorithm from Section~\ref{sec:mom}.
The original and retrieved measures produce almost identical distributions.
}
\label{fig:dd_recovery}

\end{center}
\vskip -0.2in
\end{figure*}

\begin{table*}[tb]
\centering
\begin{tabular}{l c c c c c c c c}
\toprule
Generating Measure  & $|V|$ & $m$ & $k$ & $\ell_1$ & $\ell_2$ & $p_{11}$ & $p_{12}$ & $p_{22}$ \\
\midrule
Original $\W_k$ & 5,000 & 2 & 12 & 0.5 & 0.5 & 0.73 & 0.73 & 0.73  \\
Retrieved  $\bar{\W}_{\bar{k}}$ & 5,000 & 2 & 10 &  0.0574 & 0.9425 & 0.0074 & 0.7273 & 0.6829 \\
\bottomrule
\end{tabular}
\caption{
Comparison of original measure to the measure retrieved by using the method of moments algorithm from Section~\ref{sec:mom}.
The graph features used for the method of moments were: number of edges, number of $d$-stars for $d = 2, 3, 4, 5$, and number of $t$-cliques for $t = 3, 4$.
The original generative measure is an Erd\H{o}s-R\'{e}nyi random graph model.
While the recursion depth, probabilities, and lengths vector are quite different, the retrieved measure is similar to the same Erd\H{o}s-R\'{e}nyi model (see the discussion in Section~\ref{sec:synthetic}).
}
\label{tab:ER_recovery}
\end{table*}

\begin{table*}[tb]
\centering
\begin{tabular}{l c c c c c c c c}
\toprule
Generating Measure  & $|V|$ & $m$ & $k$ & $\ell_1$ & $\ell_2$ & $p_{11}$ & $p_{12}$ & $p_{22}$ \\
\midrule
Original $\W_k$ & 6,000 & 2 & 10 & 0.25 & 0.75 & 0.59 & 0.43 & 0.78  \\
Retrieved  $\bar{\W}_{\bar{k}}$ & 6,000 & 2 & 9 &  0.2728 & 0.7272 & 0.5431 & 0.4101 & 0.7593 \\
\bottomrule
\end{tabular}
\caption{
Comparison of original measure to the measure retrieved by using the method of moments algorithm from Section~\ref{sec:mom}.
The graph features used for the method of moments were: number of edges, number of $d$-stars for $d = 2, 3, 4, 5$, and number of $t$-cliques for $t = 3, 4$.
All parameters in the retrieved measure are remarkably similar to the parameters in the original measure.
}
\label{tab:mfng_recovery}
\end{table*}

Before turning to real networks, it is important to see if our method of moments algorithm recovers the structure of graphs that are actually generated by MFNG with some measure $\W_k$.
In other words, can our method of moments identify graphs generated from our model?
There are two success metrics for recovery of the generating measure.
First, we want the method of moments to recover a measure similar to $\W_k$.
Second, even if we cannot recover the measure, we want a measure that has similar feature counts.
Our experiments in this section show that we can be successful in both metrics.
If we can recover a measure with similar moments, then the new measure will be a useful model for the old one.
This is our interest when modeling real data sets in Section~\ref{sec:learning}.

Our basic experiment is as follows:
\begin{enumerate}
\item
Construct a measure $\W_k(\m P, \ell)$ and generate a \emph{single} graph $G$ from the measure.

\item
Run the method of moment algorithm from Section~\ref{sec:mom} with $G$ using 10,000 random restarts.
Fit the moments for the following graph features: number of edges, number of $d$-stars for $d = 2, 3, 4, 5$, and number of $t$-cliques for $t = 3, 4$.
The measure given by the method of moments is denoted $\bar{W}_k(\bar{\m P}, \bar{\ell})$.

\item
To compare $\W_k(\m P, \ell)$ and $\bar{W}_k(\bar{\m P}, \bar{\ell})$,
sample 100 graphs from each measure and look at the histogram of the features that were considered by the method of moments algorithm.
\end{enumerate}

We used two different measures $\W_k$ for testing.
The first was equivalent to an Erd\H{o}s-R\'{e}nyi random graph.
This is modeled by a generating measure $\W_{k}(\m P, \ell)$
where every entry of $\m P$ is identical.
In this case, MFNG is an Erd\H{o}s-R\'{e}nyi generative model with edge probability $\m P_{11}^k$,
independent of $\ell$.
Table~\ref{tab:ER_recovery} shows the retrieved measure  $\bar{\W}_{\bar{k}}(\bar{\m P}, \bar{\ell})$ and the original Erd\H{o}s-R\'{e}nyi measure $\W_k(\m P, \ell)$.
While $\bar{\m P}$ and $\bar{\ell}$ are quite different than $\m P$ and $\ell$,
 $\bar{\W}_{\bar{k}}(\bar{\m P}, \bar{\ell})$ still represents a measure close to an Erd\H{o}s-R\'{e}nyi random graph model.
The reason is that the length vector $\ell$ is heavily skewed to the second component ($\ell_2 \approx 0.94$).
In expectation, $0.94^{\bar{k}} \approx 0.53$ of the nodes correspond to the same category at each level.
These nodes are all connected with probability $0.6829^{\bar{k}} \approx 0.022$, which is  nearly the same as the edge probability in the original Erd\H{o}s-R\'{e}nyi measure.
Figure~\ref{fig:hist_ER_recovery} shows the histograms of the features that were used in the method of moments algorithms (as well as the clustering coefficient).
The green histogram is the data for graphs sampled from $\W_k(\m P, \ell)$,
the red histogram is the same data for graphs sampled from  $\bar{\W}_{\bar{k}}(\bar{\m P}, \bar{\ell})$,
and the blue line is the feature count in the original graph $G$ used as input to the method of moments.
There is remarkable overlap between the empirical distribution of the features for  $\bar{\W}_{\bar{k}}(\bar{\m P}, \bar{\ell})$ and the distribution of the features for the original measure.

For a second experiment, we used an original measure $\W_k(\m P, \ell)$ that did not possess the uniform generative structure of Erd\H{o}s-R\'{e}nyi random graphs.
Table~\ref{tab:mfng_recovery} shows the retrieved measure and the original measure.
In this case, the method of moments identified a similar generative measure.
The parameters $\bar{k}$, $\bar{\m P}$, and $\bar{\ell}$ are remarkably similar to $k$, $\m P$, and $\ell$.
Figure~\ref{fig:hist_mfng_recovery} shows the distribution of the features in graphs sampled from the two measures.
Again, there is rather significant overlap in the empirical distributions.

Finally, we compare the degree distributions of the original and retrieved measures in Figure~\ref{fig:dd_recovery}.
The degree distributions are nearly identical.

These results show that the method of moments algorithm described in Section~\ref{sec:mom} can successfully identify MFNG instances using a single sample.
In the case when the original measure was Erd\H{o}s-R\'{e}nyi, the retrieved measure parameters looked different but the measures produced similar graphs.
When a more sophisticated MFNG was used, the retrieved measure had similar parameters \emph{and} produced similar graphs.

\subsection{Learning on real networks}
\label{sec:learning}

\begin{table*}[tb]
\centering
\begin{tabular}{l l l c c c c c c c}
\toprule
\textbf{Network}  & $m$ & features for & $|V|$ & $|E|$ & $S_2$ & $S_3$ & $S_4$ & $C_3$ & $C_4$ \\
method & & fitting \\\midrule
\textbf{Gnutella} & -- & -- & 62,586 & 147,892 & 1.57e+06 & 8.17e+06 & 4.38e+07 & 2.02e+03 & 1.6e+01 \\
MFNG MoM              & $2$ & all      & -- &  1.13 & 1.00 & 0.97 & 1.00 & 1.00 & 1.00 \\
MFNG MoM              & $3$ & all      & -- &  1.00 & 0.97 & 1.00 & 1.00 & 1.00 & 1.00 \\
MFNG MoM              & $2$ & $|E|, S_2, C_3$   & -- & 1.00 &   1.00  &   1.10 & 1.15 &   1.00 &   0.05 \\
MFNG MoM              & $3$ & $|E|, S_2, C_3$   & -- & 1.00 &   1.00  &   1.21 & 1.54 &   1.00 &   0.26\\
SKG MoM           & -- & $|E|, S_2, S_3, C_3$   & -- & 1.14 & 1.00 & 1.00 & 18.34 & 0.30 & < 0.69  \\
KronFit           & --  & --       & -- &  0.54 & 0.30 & 0.23 & 3.67 & 0.06 & $<$ 0.01 \\ \midrule
\textbf{Citation} & -- & -- & 34,546 & 42,0921 & 2.63e+07 & 1.34e+09 & 1.04e+10 & 1.28e+06 & 2.57e+06 \\
MFNG MoM & $2$ & all      & -- & 0.79 & 1.02 & 1.00 & 0.61 & 1.00 & 1.00 \\
MFNG MoM & $3$ & all      & -- & 1.00 & 1.03 & 1.00 & 1.00 & 1.00 & 1.00 \\
MFNG MoM & $2$ & $|E|, S_2, C_3$   & -- & 0.99 & 1.00 & 0.77 & 0.42 & 1.00 & 4.65 \\
MFNG MoM & $3$ & $|E|, S_2, C_3$   & -- & 1.00 & 1.00 & 0.85 & 0.60 & 1.00 & 1.08 \\
SKG MoM           & -- & $|E|, S_2, S_3, C_3$   & -- &  1.00 & 0.89 & 1.00 & 11.60 & 0.02 & < 0.01 \\
KronFit           & --  & --       & -- & 0.53 & 0.21 & 0.09 & 0.57 & $<$ 0.01 & $<$ 0.01 \\ \midrule
\textbf{Twitter} & -- & -- & 81,306 & 1,342,310 & 2.30e+08 & 6.35e+10 &  2.99e+13 & 1.31e+07 & 1.05e+08 \\
MFNG MoM & $2$ & all      & -- & 1.00 & 1.59 & 1.00 & 0.33 & 1.00 & 1.00 \\
MFNG MoM & $3$ & all      & -- & 1.00 & 1.16 & 1.00 & 1.00 & 0.89 & 1.00 \\
MFNG MoM & $2$ & $|E|, S_2, C_3$   & -- & 1.00 & 1.00 & 0.44 & 0.12 & 1.00 & 2.83 \\
MFNG MoM & $3$ & $|E|, S_2, C_3$   & -- & 1.00 & 1.00 & 0.44 & 0.11 & 1.00 & 2.71 \\
SKG MoM           & -- & $|E|, S_2, S_3, C_3$   & -- &  1.00 & 1.05 & 1.00 & 0.01 & 0.03 & $<$ 0.01 \\
KronFit            & --  & --       & -- & 0.69 & 0.30 & 0.10 & $<$ 0.01 & $<$ 0.01 & $<$ 0.01 \\ \midrule
\textbf{Facebook} & -- & -- & 4,039 & 88,234 & 9.31e+06 & 7.27e+08 & 9.71e+10 & 1.61e+06 & 3.00e+07 \\
MFNG MoM & $2$ & all      & -- & 0.96 & 1.19 & 1.00 & 0.42 & 1.00 & 1.00 \\
MFNG MoM & $3$ & all      & -- & 1.00 & 1.06 & 1.00 & 0.69 & 0.90 & 1.00 \\
MFNG MoM & $2$ & $|E|, S_2, C_3$   & -- & 0.90 & 1.00 & 0.80 & 0.34 & 1.00 & 1.88 \\
MFNG MoM & $3$ & $|E|, S_2, C_3$   & -- & 1.00 & 1.00 & 0.75 & 0.33 & 1.00 & 1.13 \\
SKG MoM  & -- & $|E|, S_2, S_3, C_3$   & -- &  1.00 & 1.03 & 1.00 & 0.19 & 0.08 & 0.03 \\
KronFit  & --  & --       & -- & 0.49 & 0.20 & 0.07 & 0.04 & 0.01 & $<$ 0.01 \\
\bottomrule
\end{tabular}
\caption{
Results of method of moments (MoM) fit to MFNG for several graphs.
Each column gives the ratio of the expected feature count to the true feature count.
$S_d$ is the number of $d$-stars in the graph, and $C_t$ is the number of $t$-cliques in the graph.
A value of 1.00 means that the moment is an exact fit to two decimal places.
In all cases, MFNG is able to fit many of the feature counts exactly in expectation.
For MFNG, we fit all feature moments listed and fitting just the number of edges, wedges, and triangles.
The SKG MoM and KronFit are included for comparison.
For these methods, $S_4$ and $C_4$ were estimated by taking the mean from 10 sample graphs
(closed-form moment formulas are not available for these feature counts).
Our MFNG MoM outperforms both KronFit and SKG MoM in fitting feature moments.
}
\label{tab:mom_fits}
\end{table*}

\begin{figure*}[tb]
\vskip 0.2in
\begin{center}
\includegraphics[width=1.72in]{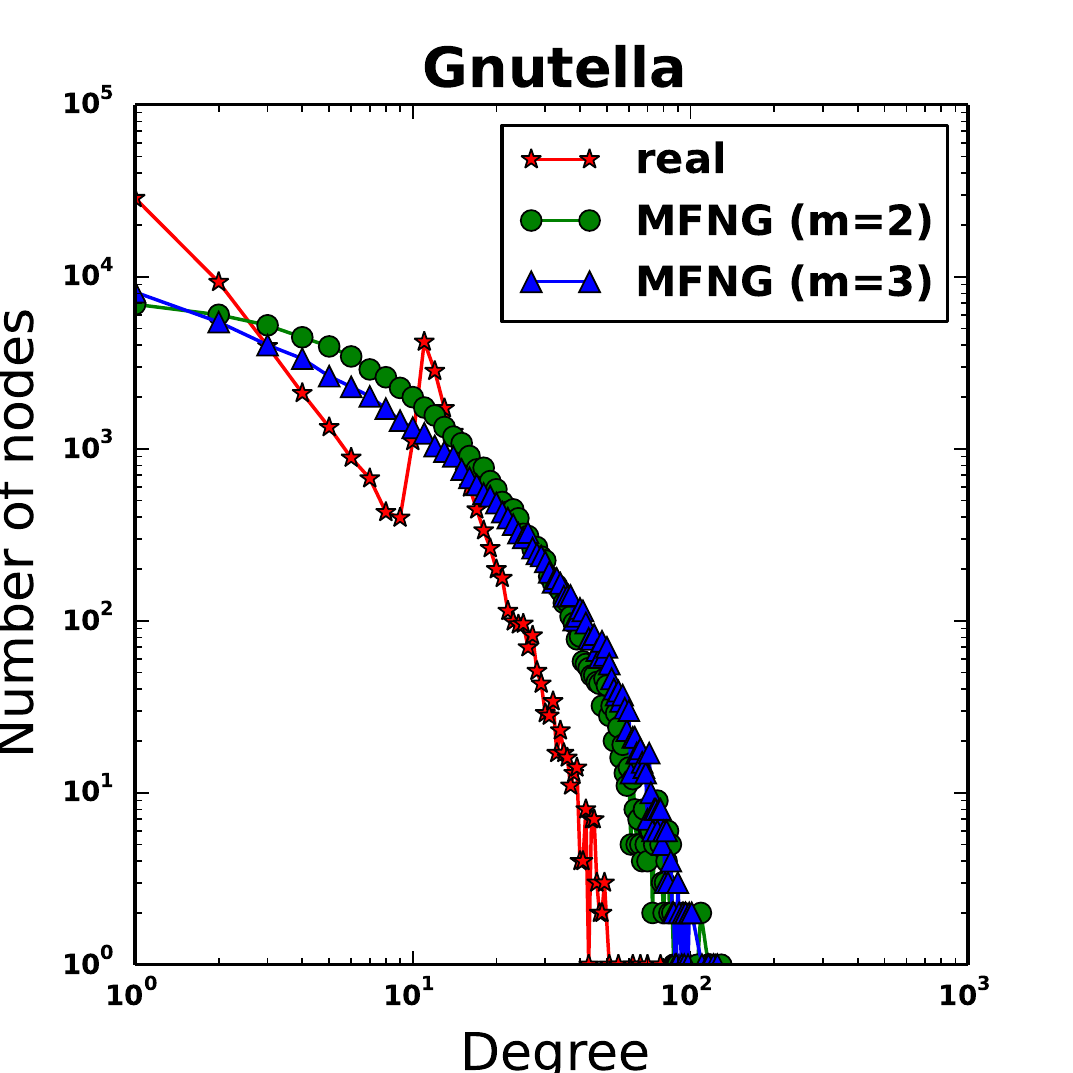}
\includegraphics[width=1.72in]{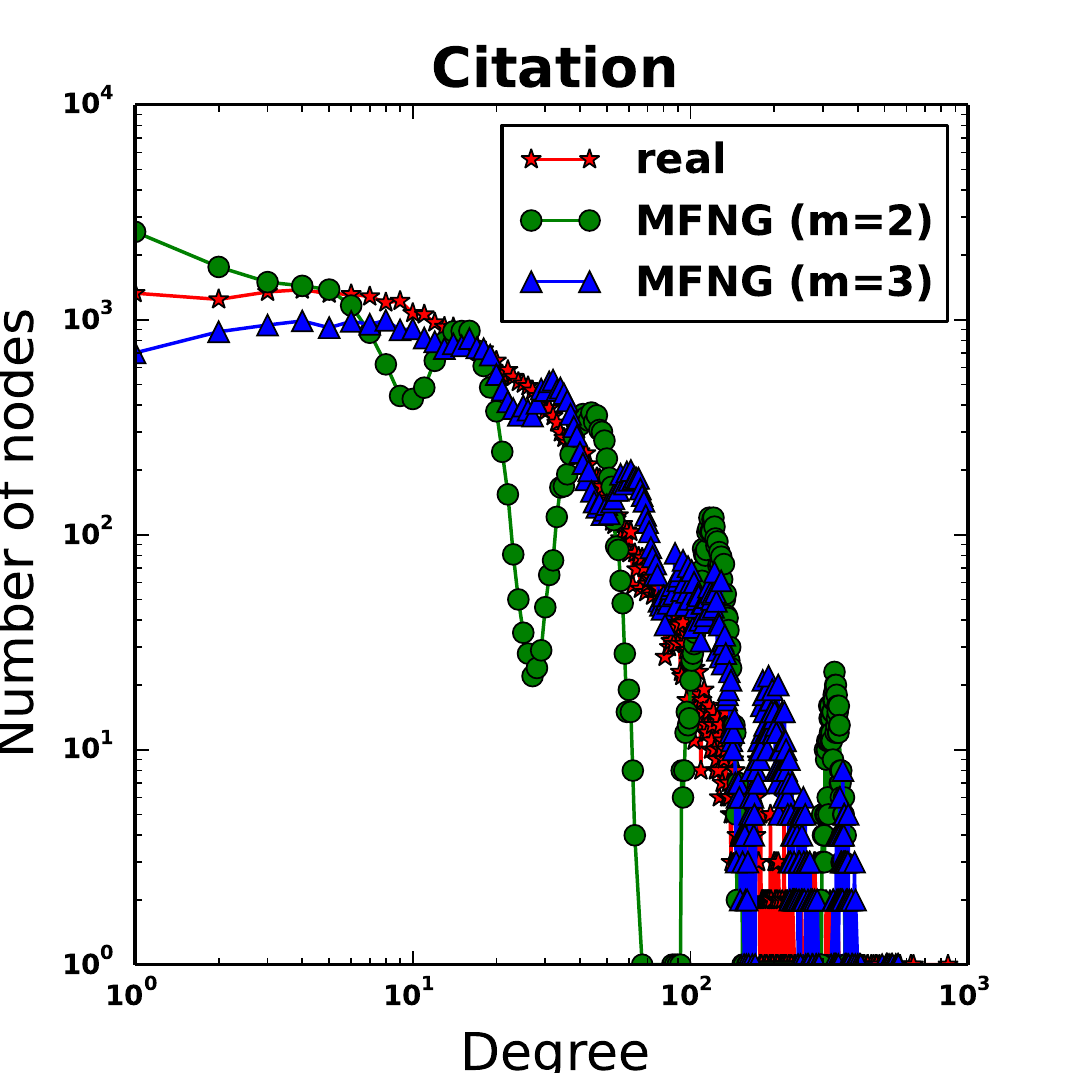}
\includegraphics[width=1.72in]{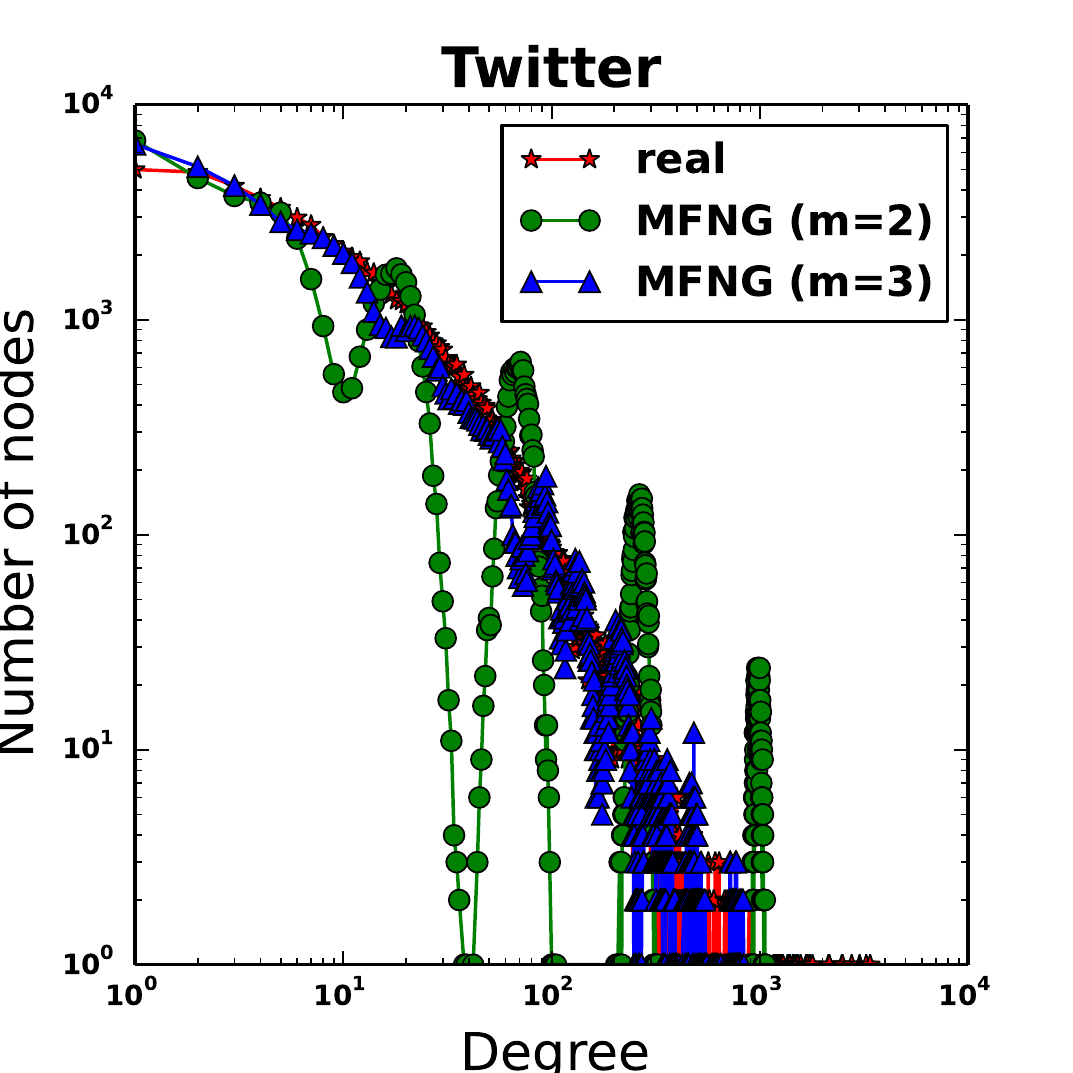}
\includegraphics[width=1.72in]{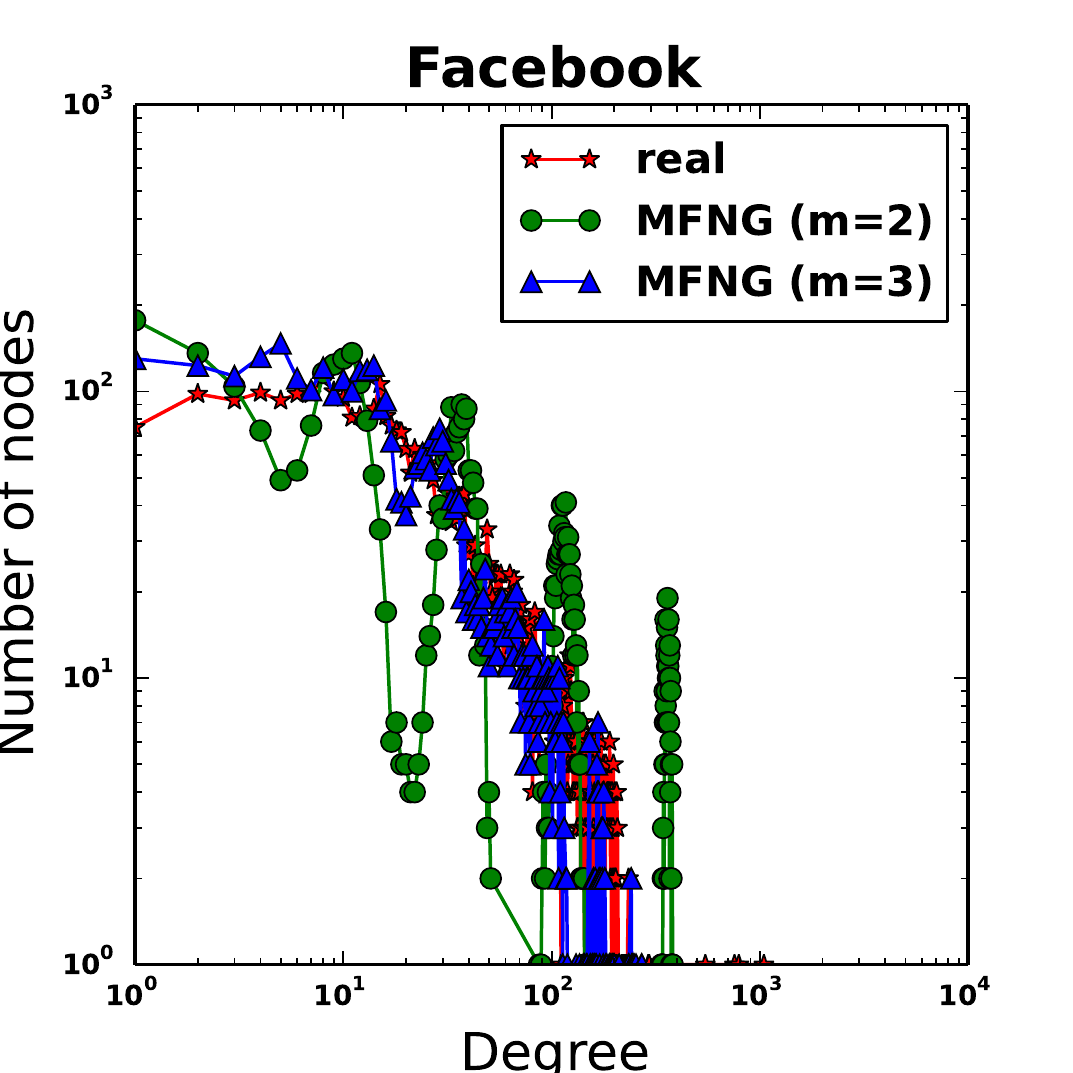}
\caption{Degree distributions for the original graphs and MFNG graphs
         with $m = 2, 3$ for several networks.
         The degree distributions of the MFNG graphs are similar
         to those of the original network, even though we only fit
         $d$-star and $t$-clique moments.
         In the Twitter, Citation, and Facebook graphs, the MFNG fit
         with $m = 2$ results in oscillating degree distributions.
         In Section~\ref{sec:noise}, we show how to add noise to dampen
         the oscillations.
         The graph samples were generated with the fast sampling algorithm in Section~\ref{sec:fast}.
        }
\label{fig:mfng_dd}
\end{center}
\vskip -0.2in
\end{figure*}

\begin{figure*}[tb]
\vskip 0.2in
\begin{center}
\includegraphics[width=2.2in]{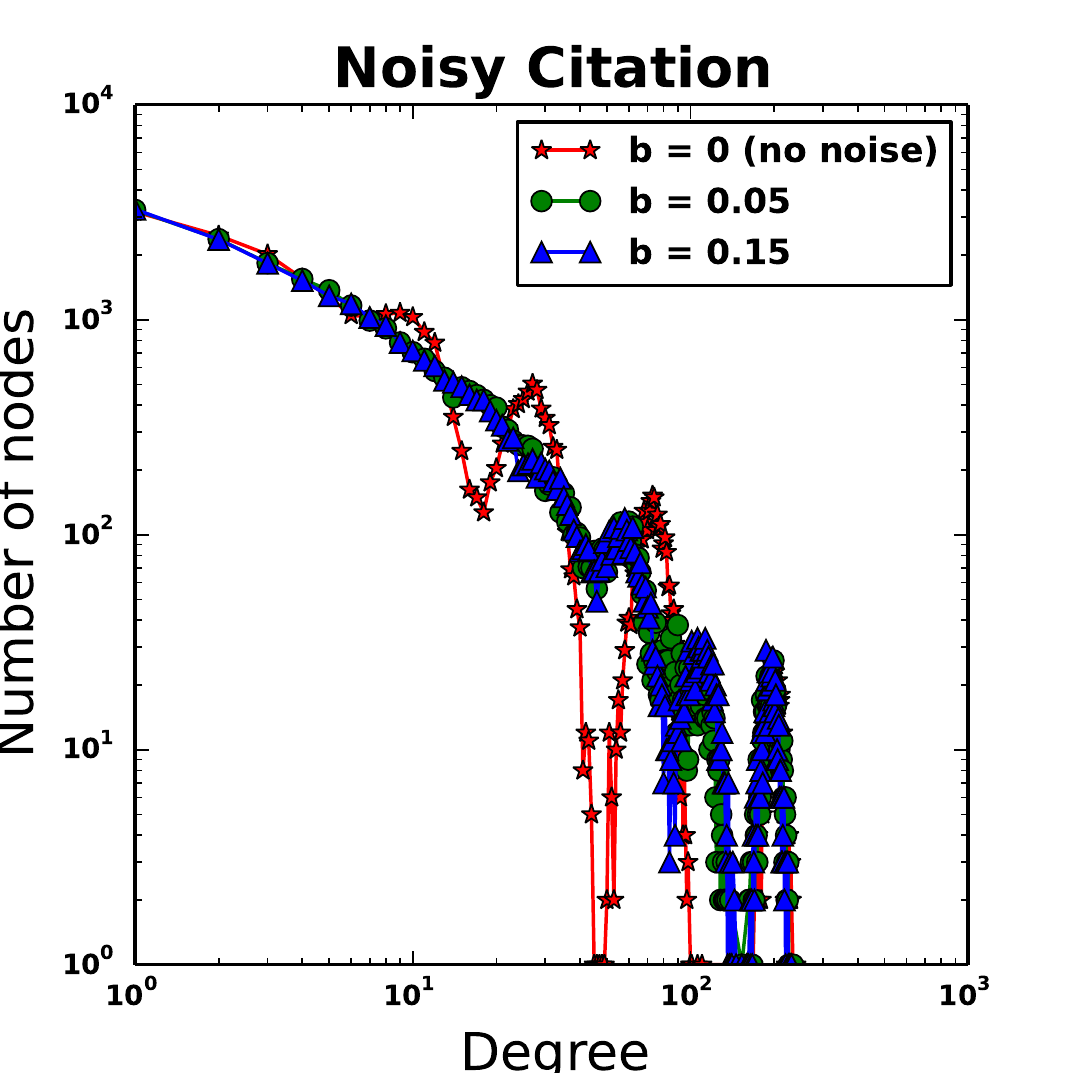}
\includegraphics[width=2.2in]{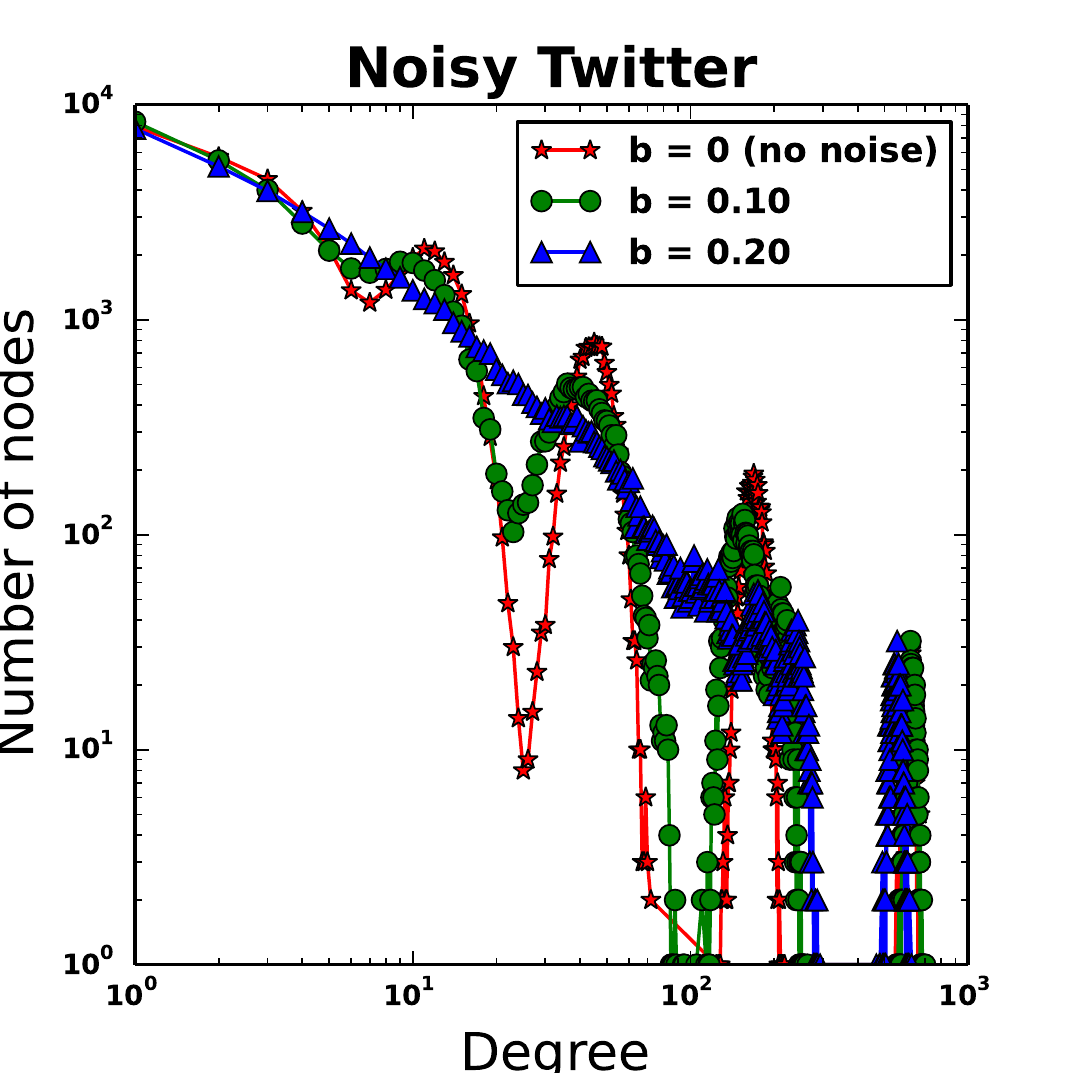}
\caption{Degree distributions for fitting the citation and Twitter networks
         to Noisy MFNG with varying degrees of noise.
         We see that adding noise dampens the oscillations in the degree
         distributions.  At the far end of the tail, it is still
         difficult to control the degree distribution.
         The graph samples were generated with the fast sampling algorithm in Section~\ref{sec:fast}.
        }
\label{fig:noisy_mfng}
\end{center}
\vskip -0.2in
\end{figure*}

We now show how the method of moments from Section~\ref{sec:mom} performs when
fitting to the following four real-world networks to MFNG:
\begin{enumerate}
\item The Gnutella graph is a network of host computers sharing files on August 31, 2012 \cite{leskovec2007graph}.
\item The Citation network is from a set of high energy physics papers from arXiv \cite{leskovec2005graphs}.
\item The Twitter network is a combination of several ego networks from the Twitter follower graph \cite{mcauley2012learning}.
\item The Facebook network is a combination of several ego networks from the Facebook friend graph \cite{mcauley2012learning}.
\end{enumerate}

All data sets are from the SNAP collection.
We use the optimization procedure described in Section~\ref{sec:mom} with 2,000 random restarts.
The features we use (the $F_i$ in Section~\ref{sec:mom})
are number of edges, wedges ($S_2$), $3$-stars ($S_3$), $4$-stars ($S_4$), triangles ($C_3$), and $4$-cliques ($C_4$).
For each network, we fit with $m = 2, 3$ and with $k = \lceil \log_m(|V|) \rceil$.
While $k$ can be arbitrary, a smaller value of $k$ leads to many nodes
belonging to the same categories and hence having the same statistical properties.
In large graphs, this causes a ``clumping'' of properties such as degree distribution near a small set of discrete values.
While smaller $k$ may be satisfactory for testing algorithms, keeping $k$ near $\log_m(|V|)$ produces more realistic graphs.
In an additional set of experiments, we only fit the number of edges, wedges, and triangles.
We also compare against KronFit and the SKG method of moments \cite{gleich2012moment}.

The results of the optimization procedure for all experiments are in Table~\ref{tab:mom_fits}.
Overall, for both $m = 2$ and $m = 3$, the method of moments can effectively match most feature counts.
The number of $4$-stars ($S_4$) was the most difficult parameter to fit.
We see that when only fitting the number of edges, wedges, and triangles,
the other feature moments can be significantly different from the original graph.
In particular, the number of $4$-cliques tends to be severely under- or over-estimated.
Although KronFit does not explicitly try to fit moments,
the results show that it severely underestimate several feature counts.
The SKG method of moments can fit three of the features,
which is consistent with results on other networks \cite{gleich2012moment}.

As mentioned in Section~\ref{sec:desired}, the clustering coefficient is three times
the ratio of the number of triangles ($3$-cliques) to the number of wedges ($2$-stars) in the graph.
The results of Table~\ref{tab:mom_fits} show that the method of moments can match
both the number of triangles and the number of wedges in expectation.
This does not make any guarantees about the \emph{ratio} of these random variables,
but the synthetic experiments (Section~\ref{sec:synthetic}) demonstrated that their variances
are not too large.
Therefore, the expectation of the ratio is near the ratio of the expectations,
and we approximately match the global clustering coefficient.

Figure~\ref{fig:mfng_dd} shows the degree distributions for the original networks
and a sample from the corresponding MFNG, using the fast sampling algorithm.
We see that, even though we only fit feature moments,
the global degree distribution is similar to the real network.
However, the MFNG degree distributions experience oscillations, especially in the case when $m = 2$.
This is a well-known issue in SKG \cite{seshadhri2013depth},
and we address this issue in Section~\ref{sec:noise}.
Finally, note that we only plot the degree distribution for a single MFNG sample.
The reason is that the samples tend to have quite similar degree distributions.
This lack of variance has been observed for SKG \cite{moreno2010tied},
and addressing this issue for MFNG is an area of future work.

\subsection{Noisy MFNG}
\label{sec:noise}

\begin{algorithm}[t]
   \caption{Noisy MFNG ($m = 2$)}
   \label{alg:noisy_mfng}
\begin{algorithmic}[1]
   \STATE {\bfseries Input:} $2 \times 2$ probability matrix $\m P$,
                             lengths vector $\ell$,
                             number of recursive levels $k$,
                             noise level $b$
   \STATE {\bfseries Output:} noisy MFNG matrix $G$
   \FOR{$i=1$ {\bfseries to} $k$}
   \STATE Sample $\mu_i \sim \text{Uniform}[-b, b]$.
   \STATE $\m P^{(i)} =
           \begin{pmatrix} \m p_{11} -\frac{2\mu_i\m p_{11}}{\m p_{11} + p_{22}}
                            & \m p_{12} + \mu_i \\
                            \m p_{21} + \mu_i
                            & \m p_{22} - \frac{2\mu_i\m p_{22}}{\m p_{11} + p_{22}}
           \end{pmatrix}$
   \STATE $\m P^{(i)} = \min(\max(\m P^{(i)}, 0), 1)$ entry-wise
   \STATE Sample $H_i \sim \W_1(\m P^{(i)}, \m l)$
   \ENDFOR
   \STATE $G := \cap_{i=1}^{k} H_i$
\end{algorithmic}
\end{algorithm}

Figure~\ref{fig:mfng_dd} shows that the graphs generated with MFNG experience
oscillations in the degree distribution.
The oscillations for the degree distribution are a well-known issue in SKG \cite{seshadhri2013depth}.
Seshadhri \emph{et al}.~present a ``Noisy SKG'' model that perturbs
the initiator matrix at each recursive level,
which dampens the oscillations.
Inspired by their work, we present a similar ``Noisy MFNG" in this section.

We first note that Figure~\ref{fig:mfng_dd} shows that using $m = 3$ results in less severe oscillations in the degree distribution.
The intuition behind this is that more categories get mixed at each recursive level, producing  a larger variety of edge probabilities.
For $m = 2$, we propose a Noisy MFNG model,
which is described in Algorithm~\ref{alg:noisy_mfng}.
The basic idea is to perturb the probability matrix slightly at each level.
In the context of Lemma~\ref{thm:coupling}, this means that the noisy MFNG
graph is the intersection of several graphs generated from slightly
different probability matrices.
The fast generation method still works in this case---a different matrix at each level determines the categories instead of one single matrix.
The probability perturbations are analogous to those performed by Seshadhri \emph{et al}.
We test Noisy MFNG on the citation network and the Twitter ego network,
and the results are in Figure~\ref{fig:noisy_mfng}.
The graphs are sampled using the fast sampling algorithm.
We see that increasing the noise significantly dampens the degree distribution.
However, the far end of the tail still experiences some oscillations.

\section{Discussion}
\label{sec:conclusion}
We have shown that the multifractal graph paradigm is well suited to model and capture the properties of real-world networks
by building on the work of Palla et al. \cite{palla2010multifractal} and incorporating several ideas from SKG.
The foundation of our theoretical work is Theorem~\ref{thm:main}, which has opened the door to quick evaluation of the
expected value of a number of important counts of subgraphs, such as $d$-stars and $t$-cliques.
Combined with standard optimization routines, we are able to fit large graphs easily.
Our method of moments algorithm identifies synthetically generated MFNG and also produces close fits for real-world networks.
It is quite amazing how fitting a few `local' properties leads to a generator that fits the overall structure of graphs well.

This would not be too useful if we were not able to also generate multifractal graphs of the same scale.
For this, we presented a fast heuristic approximation algorithm that generates such graphs in $\mathcal{O}(|E|\log |V|)$ complexity,
rather than the naive $\mathcal{O}(|V|^2)$ algorithm.
Since many real-world networks are sparse, this is a significant improvement.

Future work includes the development of approximation formulas for the moments of global properties like graph diameter and a more tailored approach in the optimization routines for the fitting.
A pressing issue is the theory behind the fast generation method.
While the generation tends to produce similar graphs to the naive generation in practice,
we want to prove that the approximation is good.
Furthermore, it is possible to improve the generation further by considering a parallel implementation.
Lastly, it would be interesting to do a theoretical analysis of the oscillatory degree distribution,
similar in spirit to \cite{seshadhri2013depth}.

\section{Acknowledgements}
We thank David Gleich and Victor Minden for helpful discussions. 
Austin R. Benson is supported by an Office of Technology Licensing Stanford Graduate Fellowship. Carlos Riquelme is supported by a DARPA grant research assistantship under the supervision of Prof.\ Ramesh Johari.
Sven Schmit is supported by a Prins Bernhard Cultuurfonds Fellowship.

\bibliographystyle{abbrv}
\bibliography{mfng}

\begin{thebibliography}{10}

\bibitem{gleich2012moment}
D.~F. Gleich and A.~B. Owen.
\newblock Moment-based estimation of stochastic kronecker graph parameters.
\newblock {\em Internet Mathematics}, 8(3):232--256, 2012.

\bibitem{kim2010multiplicative}
M.~Kim and J.~Leskovec.
\newblock Multiplicative attribute graph model of real-world networks.
\newblock In {\em Algorithms and Models for the Web-Graph}, pages 62--73.
  Springer, 2010.

\bibitem{leskovec2010kronecker}
J.~Leskovec, D.~Chakrabarti, J.~Kleinberg, C.~Faloutsos, and Z.~Ghahramani.
\newblock Kronecker graphs: An approach to modeling networks.
\newblock {\em The Journal of Machine Learning Research}, 11:985--1042, 2010.

\bibitem{leskovec2007scalable}
J.~Leskovec and C.~Faloutsos.
\newblock Scalable modeling of real graphs using kronecker multiplication.
\newblock In {\em Proceedings of the 24th international conference on Machine
  learning}, pages 497--504. ACM, 2007.

\bibitem{leskovec2005graphs}
J.~Leskovec, J.~Kleinberg, and C.~Faloutsos.
\newblock Graphs over time: densification laws, shrinking diameters and
  possible explanations.
\newblock In {\em Proceedings of the eleventh ACM SIGKDD international
  conference on Knowledge discovery in data mining}, pages 177--187. ACM, 2005.

\bibitem{leskovec2007graph}
J.~Leskovec, J.~Kleinberg, and C.~Faloutsos.
\newblock Graph evolution: Densification and shrinking diameters.
\newblock {\em ACM Transactions on Knowledge Discovery from Data (TKDD)},
  1(1):2, 2007.

\bibitem{mcauley2012learning}
J.~McAuley and J.~Leskovec.
\newblock Learning to discover social circles in ego networks.
\newblock In {\em Advances in Neural Information Processing Systems 25}, pages
  548--556, 2012.

\bibitem{moreno2010tied}
S.~Moreno, S.~Kirshner, J.~Neville, and S.~Vishwanathan.
\newblock Tied kronecker product graph models to capture variance in network
  populations.
\newblock In {\em Communication, Control, and Computing (Allerton), 2010 48th
  Annual Allerton Conference on}, pages 1137--1144, Sept 2010.

\bibitem{moreno2013network}
S.~Moreno and J.~Neville.
\newblock Network hypothesis testing using mixed kronecker product graph
  models.
\newblock In {\em ICDM}, pages 1163--1168, 2013.

\bibitem{moreno2013learning}
S.~I. Moreno, J.~Neville, and S.~Kirshner.
\newblock Learning mixed kronecker product graph models with simulated method
  of moments.
\newblock In {\em Proceedings of the 19th ACM SIGKDD international conference
  on Knowledge discovery and data mining}, pages 1052--1060. ACM, 2013.

\bibitem{murphy2010introducing}
R.~C. Murphy, K.~B. Wheeler, B.~W. Barrett, and J.~A. Ang.
\newblock Introducing the graph 500.
\newblock {\em Cray User’s Group (CUG)}, 2010.

\bibitem{palla2010multifractal}
G.~Palla, L.~Lov{\'a}sz, and T.~Vicsek.
\newblock Multifractal network generator.
\newblock {\em Proceedings of the National Academy of Sciences},
  107(17):7640--7645, 2010.

\bibitem{palla2011rotated}
G.~Palla, P.~Pollner, and T.~Vicsek.
\newblock Rotated multifractal network generator.
\newblock {\em Journal of Statistical Mechanics: Theory and Experiment},
  2011(02):P02003, 2011.

\bibitem{seshadhri2012community}
C.~Seshadhri, T.~G. Kolda, and A.~Pinar.
\newblock Community structure and scale-free collections of
  erd{\H{o}}s-r{\'e}nyi graphs.
\newblock {\em Physical Review E}, 85(5):056109, 2012.

\bibitem{seshadhri2013depth}
C.~Seshadhri, A.~Pinar, and T.~G. Kolda.
\newblock An in-depth analysis of stochastic kronecker graphs.
\newblock {\em Journal of the ACM (JACM)}, 60(2):13, 2013.

\end{thebibliography}
%

\end{document}